\documentclass[a4paper]{article}

\usepackage[utf8]{inputenc}   
\usepackage[T1]{fontenc}      
\usepackage[english]{babel}

\usepackage{amsfonts}
\usepackage{mathrsfs}
\usepackage{dsfont}
\usepackage{amssymb}        
\usepackage{amsmath}
\usepackage{graphicx}
\usepackage{float}
\usepackage[a4paper]{geometry}
\usepackage{verbatim}

\usepackage{amsthm}
\usepackage{mathrsfs}
\DeclareMathOperator{\Tr}{Tr}
\DeclareMathOperator{\supp}{Supp}

\newtheorem{Prop}{Proposition}[section]
\newtheorem{Theo}{Theorem}[section]
\newtheorem{Def}{Definition}[section]
\newtheorem{Lem}{Lemma}[section]
\newtheorem{Cor}{Corollary}[section]

\author{Raphael Ducatez\footnote{CEREMADE, UMR CNRS 7534, Université Paris-Dauphine, Place de Lattre de Tassigny, 75775 Paris cedex~16, FRANCE. email : ducatez@ceremade.dauphine.fr}}
\title{Anderson localisation for infinitely many interacting particles in Hartree-Fock theory}                              
                              
\begin{document}

\maketitle

\begin{abstract}
We prove the occurrence of Anderson localisation for a system of infinitely many particles interacting with a short range potential, within the ground state Hartree-Fock approximation. We assume that the particles hop on a discrete lattice and that they are submitted to an external periodic potential which creates a gap in the non-interacting one particle Hamiltonian. We also assume that the interaction is weak enough to preserve a gap. We prove that the mean-field operator has exponentially localised eigenvectors, either on its whole spectrum or at the edges of its bands, depending on the strength of the disorder.  

\bigskip


\smallskip

\end{abstract}

\tableofcontents 
\section{Introduction}

In 1958, the physicist P.W.~Anderson predicted that, in a random medium, diffusion could disappear and waves stay localised \cite{PhysRev.109.1492}.
Since then, \emph{Anderson localisation} has played an important role to explain several properties of physical systems. 

There are many works on the mathematical side, starting with the simpler one dimension  case in the beginning of the 80s. In dimension 3, the first proof of Anderson localisation was provided by Fröhlich and Spencer \cite{frohlich1983} who developed a method now called \emph{multiscale analysis}. There are now other approaches which include for instance the \emph{fractional moment method} proposed by Aizenman and Molchanov \cite{aizenman1993} or the techniques proposed by Imbrie \cite{2014arXiv1406.2957I}.

All these works are restricted to the one wave problem, which is well adapted to optics and acoustics. In condensed matter, the interaction between the particles is expected to play an important role and in this case one speaks of \emph{many-body localisation}. This question is not fully understood and is very actively studied in physics \cite{2006AnPhy.321.1126B,BasAleAlt-07,2009CMaPh.290..903A,PhysRevLett.95.206603}. A typical system of interest is a crystal composed of quantum electrons and classical nuclei placed on a random perturbation of a perfect lattice \cite{BlaLew-12,lahbabi:hal-00797094}. This is a very complicated system since the Coulomb interaction is long range and there are of the order of $10^{23}$ particles, usually mathematically treated as an infinite number. Due to screening effects, the Coulomb potential is often replaced by an effective short range interaction and this is what we are going to do in this work.

\emph{Finite} interacting systems have been recently considered in several works  
\cite{2009MPAG...12..117C,2009CMaPh.290..903A,KliBootstap,2015RvMaP..2750010F} but infinite systems have not been studied thoroughly. One should however mention a series of works on superfluidity and Bose-Einstein condensation in the Lieb-Liniger 1D Bose gas in a strong random potential \cite{SeiYngZag-12,KonMosSeiYng-14}, and the very recent article of Seiringer and Warzel \cite{2015arXiv151205282S} on the  1D Tonks-Girardeau gas.

It is very hard to deal with the exact interacting Schödinger problem for an infinite system. A useful and widespread approximation is the Hartree-Fock model, where the particles are treated as independent objects but see a field which depends on their own states and is then self-consistently optimised. For random systems this model has been recently introduced by E.~Cancès, S.~Lahbabi and M.~Lewin in \cite{Cances2013241}. The authors were able to construct a solution of the random nonlinear Hartree-Fock equation for an infinite system of fermions with short range interaction, but the phenomenon of Anderson localisation was only investigated numerically in \cite{lahbabi:tel-00873213}. 

In this work, we complete this program in the case of a discrete system with a periodic background. In short, we show that the unique solution of the Hartree-Fock equation
\begin{eqnarray}\label{gammaSystem0}
\begin{cases}H_{min}= -\Delta+V_0+V_\omega+W\ast \gamma(x,x)-W(x-y)\gamma(x,y)\\
\gamma = \mathds{1}_{<\mu}(H_{min}),
\end{cases} 
\end{eqnarray}
provides a mean-field operator $H_{min}$ which has exponentially localised eigenvectors, either on its whole spectrum or at the edges of its bands. Our main assumption is that the periodic potential $V_0$ is sufficiently strong to create a gap for $-\Delta+V_0$ and that the random part $V_\omega$ as well as the interaction $W$ do not alter this gap. We take the chemical potential $\mu$ in this gap. Our argument is to adapt the well known multiscale analysis in order to include the nonlinear terms. 

The paper is organised as follows. In the next section, we properly define the discrete Hartree-Fock model and state our main localisation results. Sections \ref{HamiltonianConstruction} to \ref{SMulti} are devoted to the proof of our results. Then in section \ref{Snumeric} we present some simple one dimensional numerical simulations in order to illustrate our findings. We also look at some cases which are not covered by our theorems and find, for instance, an interesting delocalisation phenomenon at the Fermi level in the absence of a gap, which deserves further investigations.

\section{Model and results}
We consider a system of fermions on a subset $\Lambda$ of the lattice $\mathbb{Z}^d$ and which are submitted to an external potential $V$. Without interactions the system is described by the one-particle Hamiltonian \begin{equation}
H^\Lambda=-\Delta+V .
\end{equation} In the canonical basis $(\delta_x)_{x\in\Lambda}$ of $\ell^2(\Lambda)$, the operator $H^\Lambda$ is defined by 
\begin{equation}\label{Laplacian}
\forall x,y \in \Lambda, \qquad \big(\delta_x ,H^\Lambda \delta_y \big) = \begin{cases}V(x) & \text{ if $x=y$,} \\ 1 & \text{ if $|x-y|=1$,} \\ 0 & \text{otherwise.} \end{cases}
\end{equation}
When $\Lambda$ is the whole lattice $\mathbb{Z}^d$ we will use the simpler notation $H=H^{\mathbb{Z}^d}$. Our theorems will actually hold for finite as well as for infinite domains. Since $H$ will later be perturbed by a non linear term describing the interactions between the particles, we will often call $H$ the linear part.

 The potential $V$ describes a crystal lattice which is randomly perturbed. It is therefore assumed to be of the form \begin{equation}
V=V_0+V_\omega ,
\end{equation} where $V_0$ is a periodic potential (of an arbitrary period) and $V_\omega$ is a random potential. We use the Anderson tight binding model where the value of $V_\omega$ is chosen independently at each site of $\Lambda$ with the same random probability law $\mathbb{P}$: 
\begin{equation}
V_\omega=\sum_{i\in \Lambda} v_i(\omega)\delta_i .
\end{equation}
Here the $v_i$ are iid random variables.
 
Under suitable regularity assumptions on $\mathbb{P}$, it is well known that $H$ displays Anderson localisation. This means that there exist small intervals at the edges of the spectral bands where the spectrum is pure point with exponentially decaying eigenvectors (``Lifshitz tail''). Moreover, if the random potential $V_\omega$ is strong enough (compared to $-\Delta$) then the whole spectrum is pure point with localised eigenvectors. 

We will show that if the interactions between the particles are small enough then the same holds for the interacting model. Before introducing interactions, we first discuss the precise assumptions that we will use for $\mathbb{P}$ and $V_0$. 

\begin{description}
\item[(A1) Regularity of $\mathbb{P}$.] We assume that $\mathbb{P}$ has a bounded support and that it has a density $\rho$ with respect to the Lebesgue measure which is Lipschitz by part. 
\end{description}

We think that this assumption can be weakened in several possible ways but we will not discuss this for the sake of simplicity. With no loss of generality, we can assume that $0\in \supp(\mathbb{P})$.

\begin{description}
\item[(A2) Gap.] We assume that $H$ has a gap $[a,b]$ in its (deterministic) spectrum.
\end{description}

This assumption implies that the periodic potential $V_0$ is not constant and strong enough. Indeed the spectrum of $H$ is almost surely equal to  
\begin{equation*}
\sigma(H )=\sigma(-\Delta+V_0)+\textrm{supp}(\mathbb{P}),
\end{equation*}
  when the support of $\mathbb{P}$ is an interval. In particular it is enough to assume that $-\Delta+V_0$ has a gap of size $G_0 > 2 |\textrm{supp}(\mathbb{P})|$. When the domain $\Lambda$ is large enough then $H^\Lambda$ will have a gap as well. We call the size of the gap $G=b-a$, and choose a chemical potential $\mu$ in the middle $\mu=(a+b)/2$. 

Now we turn to the definition of the interacting model. We assume that the interaction is translation-invariant and decays fast enough.
\begin{description}
\item[(A3) Short range interaction.] We assume that there exists $\nu>0$ such that \begin{equation}
|W(x-y)|\leq Ce^{-\nu |x-y|}
\end{equation}  
\end{description}

In the Hartree-Fock model (see for example \cite{Bach-92,BacLieSol-94,LieSim-77,Lions-87, Solovej-03, BacLieLosSol-94}) the system in the domain $\Lambda$ is completely described by its one-particle density matrix which is an orthogonal projection $\gamma$ on $\ell^2(\Lambda)$. More precisely, for  a finite system, the rank-$N$ projection $ \gamma=\sum_{i=1}^N |\phi_i\rangle \langle\phi_i|$ corresponds to the $N$-particle wave function $\Psi$,
\begin{equation}
\Psi(x_1,x_2...,x_N)=\frac{1}{\sqrt{N!}} \det\big( \phi_i(x_j) \big),
\end{equation}
called a Slater determinant. In a finite domain $\Lambda$, the many-body energy of this state is 
\begin{multline}
\left( \Psi ,\Big[\sum_i^N -\Delta_i+V(x_i) + \frac{1}{2}\sum_{i\neq j} W(x_i-x_j)\Big] \Psi \right) \\ = \Tr(H^\Lambda \gamma)+\frac{1}{2}\sum_{x,y \in \Lambda} W(x-y)\gamma(x,x)\gamma(y,y)-\frac{1}{2}\sum_{x,y\in \Lambda} W(x-y)|\gamma(x,y)|^2.
\end{multline}
We define the effective interaction $A_{\textrm{eff}}$ by 
\begin{eqnarray}
A_{\textrm{eff}}(\gamma)(x,y) = \big[\sum_{n} W(n-y)\gamma(n,n) \delta_{x=y}\big]- W(x-y)\gamma(x,y)
\end{eqnarray}
for any $\gamma$ positive bounded operator. It is the derivative of the interacting part of the energy.

Any minimiser $\gamma$ of the energy, with fixed rank $N$, is a solution of the non-linear equation which involve an effective Hamiltonian.
\begin{eqnarray}\label{gammaSystem}
\begin{cases}H_{min}= -\Delta+V+A_{\textrm{eff}}(\gamma)\\
\gamma = \mathds{1}_{<\mu'}(H_{min}),
\end{cases} 
\end{eqnarray}
In other words $\gamma$ is the projection on the $N$ first eigenfunctions of $H_{\min}$ which itself depends of $\gamma$. To be more precise the equation \eqref{gammaSystem} hold under the condition that $\lambda_{N+1}(H_{min})>\lambda_N(H_{min})$ which is known to be automatically satisfied when $W>0$ \cite{BacLieLosSol-94,BacLieSol-94}.

\paragraph{} We will first show that our problem is well defined for a bounded domain $\Lambda$ as well as for the infinite domain $\Lambda=\mathbb{Z}^d$.

\begin{Theo}[HF ground states for infinitely many particles] \label{TheoExistence} Let $V'$ be a bounded function such that $-\Delta+V'$, defined on a subset $\Lambda$ of $\mathbb{Z}^d$, has a gap $[a',b']$ in its spectrum, with $G'=b'-a'$ and $\mu '=(b'+a')/2$. If $||W||_{\ell^1}< G'/6$, then there exists a unique solution of the system $\gamma =  \mathds{1}_{\leq \mu'}(-\Delta+V'+A_{\textrm{eff}}(\gamma))$. If $\Lambda$ is finite, the trace is preserved:
\begin{equation*}
\Tr(\mathds{1}_{\leq \mu'}(-\Delta+V'+A_{\textrm{eff}}(\gamma))=\Tr(\mathds{1}_{\leq \mu'}(-\Delta+V')) 
\end{equation*}
 and, because of uniqueness, our solution is as well the unique minimiser of the energy among all Hartree-Fock states with fixed particles number $N=\Tr\big(\mathds{1}_{\leq \mu'}(-\Delta+V')\big)$.

\end{Theo}

This result will be shown in Section \ref{HamiltonianConstruction}. The continuous case is more complicated and has been studied in \cite{lahbabi:hal-00797094}. When $\Lambda$ is a large cube and $V'=V=V_0+V_\omega$ then the number of particles $\Tr(\mathds{1}_{\leq \mu'}(-\Delta+V'))$ is proportional to the volume, which is the interesting physical case. For $\Lambda=\mathbb{Z}^d$, $H_{min}(\omega)$ exists and its spectrum, as in the linear model, does not depend on $\omega$ almost surely. The reason is that $H_{min}$ is stationary with respect to space translations, which follows from the uniqueness in the theorem. 

The main result of our paper is the following:

\begin{Theo}[Anderson localisation in Hartree-Fock theory]\label{TheoMain}
Under the assumptions (A1), (A2) and (A3), then the following holds:
\begin{enumerate} 
 \item There is $\epsilon$ such that if $||W||_{\ell^1}<\epsilon$ then there are small intervals at the edges of the bands of the spectrum of $H_{min}$, where the spectrum is pure point with exponentially decaying eigenvectors.
 \item In presence of strong disorder, meaning that $||\rho||_{\infty}+||\rho'||_{\infty}$ small enough, and if $||W||_{\ell^1}<G/6$, then the whole spectrum of $H_{min}$ is pure point and its eigenvectors are exponentially decaying in space.
\end{enumerate}
\end{Theo}

This theorem will be shown by using the multiscale analysis on $H_{min}$. We will proceed in three step.

\paragraph{Step 1.} We show that because of the gap, $\gamma_{min}$ only depends locally on the potential. From a practical point of view, in order to know how $\gamma_{min}$ looks like in a box of size $L$ after solving the minimising problem for $\mathbb{Z}^d$, solving the minimising problem for a box of size $2L$ will be enough to have a very good approximation. In a more mathematical formulation:

\begin{Theo}[Locality]  \label{TheoLocal} We assume that (A2) and (A3) hold. There exist $a>0$ , $\nu>0$ and $C>0$ such that, if $||W||_{\ell^1} < a G $, then for any modification of the potential $V_\omega \rightarrow V_\omega+\delta V$ so that $V_\omega+\delta V \in \text{supp}(\mathbb{P})$, the change induced to the minimising projector given by Theorem \ref{TheoExistence} satisfies 
\begin{equation} \label{ELocal}
\sup_{y_2\in\Lambda} \big|\gamma_{min}(V+\delta V)(y_1, y_2)-\gamma_{min}(V)(y_1,y_2)\big| \leq  ||\delta V|| C e^{-\nu d(y_1,\textrm{supp}(\delta V))}.
\end{equation}
Here $d(y_1,\textrm{supp}(\delta V))$ is the distance between $y_1$ and the support of the perturbation $\delta V$. The constants do not depend on the choice of the domain $\Lambda$.
\end{Theo}

This theorem will be proved in Section \ref{SLocal}.

\paragraph{Step 2.} We then show in section \ref{SWegner} a kind of Wegner estimate. We denote by : $(H_{min})_{|\Lambda}$ the submatrices $\mathds{1}_{\Lambda}H_{min}\mathds{1}_{\Lambda}$, restricted to $\ell^2(\Lambda)$ where $\Lambda$ is a finite cube in $\mathbb{Z}^d$.
\begin{Theo}[Wegner estimate] \label{TheoWegner}
Assuming (A1), (A2), (A3), there exists $a>0$ such that if $||W||_{\ell^1}<aG$ then there exists a constant $C$ so that 
\begin{equation}\label{EWegner}
\mathbb{P}\big[d(\sigma[(H_{min})_{|\Lambda}], \lambda)<\epsilon\big] \leq C |\Lambda| \sqrt{\epsilon}\Big(|\supp(\rho)|^{-1/2}+||\rho||_\infty | \supp(\rho)|^{1/2}+ ||\rho'||_\infty |\supp(\rho)|^{3/2}\Big),
\end{equation}
for any $\lambda\in \mathbb{C}$.  
\end{Theo}
This result says that there is no arbitrary small interval where we can find for sure an eigenvalue of $(H_{min})_{|\Lambda}$.

\paragraph{Step 3.} In the last step we perform the multiscale analysis. This is explained in Section \ref{SMulti}.

\section{Construction of the mean-field Hamiltonian: proof of Theorem \ref{TheoExistence}}

\label{HamiltonianConstruction}

The aim of this  section is to prove that our operator $H_{min}$ is well defined. We will show that finding the unique solution $\gamma_{min}$ in presence of a gap can be done using a fixed point lemma. 

In this subsection we solve our system 
\begin{equation} \label{Esyst}
\begin{cases}H_{min}= -\Delta+V'+A_{\textrm{eff}}(\gamma)\\
\gamma = \mathds{1}_{<\mu'}(H_{min})
\end{cases} 
\end{equation}
under the assumption that $\mu'$ is inside a gap $[a',b']$ in the spectrum, where $\mu'=(a'+b')/2$. We introduce $G'=b'-a'$.

Let $\mathcal{C}$ be a loop in the complex plane surrounding a part $I$ of the spectrum of a operator $H$. We make the assumption that the loop does not cross the spectrum (which implies that there exist gaps above and below $I$). Then

\begin{equation}
\mathds{1}_{I}(H) = \frac{1}{2i\pi}\oint_{\mathcal{C}} (H-z)^{-1} dz
\end{equation} is the projector on the spectral subspace associated with $I$.

Let us define an application that gives this projector.
\begin{Def}[Fixed point map] Let $\mathcal{C}$ be a fixed loop in the complex plane.  For all $\gamma$ orthogonal projector and $H_{\textrm{eff}}(\gamma)=-\Delta+V'+A_{\textrm{eff}}(\gamma)$ , we define 
\begin{equation}
F(\gamma) = \frac{1}{2i\pi}\oint_{\mathcal{C}} (H_{\textrm{eff}}(\gamma)-z)^{-1} dz. 
\end{equation}
\end{Def}
This application enable us to reformulate our system \eqref{Esyst} as
\begin{equation}\label{EFixePoint}
F(\gamma)=\gamma,
\end{equation}
where the loop $\mathcal{C}$ crosses the real axis at $\mu'$. Recall that $H_{min}$ is bounded so that we can always enclose all of its spectrum below $\mu'$. Because $||A_{\textrm{eff}}||$ is bounded by $2||W||_{\ell^1}$, we always have
\begin{equation}
d\Big(\sigma \big( H_{\textrm{eff}}(\gamma)\big),\mu'\Big)\geq \frac{b'-a'}{2}-2||W||_{\ell^1}.
\end{equation}
So $(b'-a')/2 > 2||W||_{\ell^1}$ is enough to ensure that $\mathcal{C}$ never crosses the spectrum of $H_{\textrm{eff}}$ and $F$ is always well defined.
In order to solve \eqref{EFixePoint} we will show that if $||W||_{\ell^1}$ is small enough then $F$ is a contraction. 

\begin{proof}[Proof of Theorem \ref{TheoExistence}]
  Let $\gamma_1$ and $\gamma_2$ be two orthogonal projectors. Then we have 
\begin{eqnarray*}
 F(\gamma_1)-F(\gamma_2) & = & -\frac{1}{2i\pi}\oint_{\mathcal{C}} (H_{\textrm{eff}}(\gamma_2)-z)^{-1}(H_{\textrm{eff}}(\gamma_1)-H_{\textrm{eff}}(\gamma_2))(H_{\textrm{eff}}(\gamma_1)-z)^{-1} dz \\ & = & 
 -\frac{1}{2i\pi}\oint_{\mathcal{C}} (H_{\textrm{eff}}(\gamma_2)-z)^{-1}(A_{\textrm{eff}}(\gamma_1)-A_{\textrm{eff}}(\gamma_2)) (H_{\textrm{eff}}(\gamma_1)-z)^{-1} dz. 
\end{eqnarray*} 
The map $F$ does not depend on the choice of the surrounding loop provided it encloses the appropriate part of the spectrum. Expending it continuously to infinity, we can replace it in this formula by $\mathcal{C}=\mu'+i\mathbb{R}$ and write $z=\mu'+is$. We estimate  $A_{\textrm{eff}}(\gamma_1)-A_{\textrm{eff}}(\gamma_2)=A_{\textrm{eff}}(\gamma_1-\gamma_2)$ with $ 2 ||W||_{\ell^1}||\gamma_1-\gamma_2||$ and $(H_{\textrm{eff}}(\gamma_2)-\mu'+is)^{-1} \leq \big((G'/2-2||W||_{\ell^1})^2+s^2 \big)^{-1/2}$. Therefore
\begin{equation}
 ||F(\gamma_1)-F(\gamma_2)|| \leq \frac{||W||_{\ell^1}||\gamma_1-\gamma_2||}{\pi} \int_{\mathbb{R}} \big((G'/2-2||W||_{\ell^1})^2+s^2\big)^{-1}ds ,
\end{equation} so that
\begin{equation}\label{PointFixe}
 ||F(\gamma_1)-F(\gamma_2)|| \leq \frac{||W||_{\ell^1}||\gamma_1-\gamma_2||}{(G'/2-2||W||_{\ell^1})}.
\end{equation}
Now, if $||W||_{L^1}$ is smaller than $G'/6$, then $F$ is contracting, so it has a unique fixed point.
\end{proof}

This concludes the first section, we have shown that in the presence of a gap, \eqref{Esyst} has always a unique solution.

\section{Local influence: proof of Theorem \ref{TheoLocal}}

\label{SLocal}

 The aim of this section is to show that under hypothesis (A2) and (A3) the random potential in a domain $\Lambda_2$ will only have a very small influence on $A_{\textrm{eff}}$ in $\Lambda_1$ if $\Lambda_2$ is far enough from $\Lambda_1$.
This implies a weak form of independence between the submatrices $(H_{min})_{|\Lambda_1}$ and $(H_{min})_{|\Lambda_2}$ which is necessary for the multiscale analysis. The key tool is a Combes-Thomas estimate.

\subsection{Combes-Thomas estimate}

Because we want to use it for more general operators than just the Laplacian, we have written again the details of the proof.

\begin{Def}[Exponential off-diagonal decay operator]
We will say that an operator $K$ on $L^2(\Lambda)$ has exponential off-diagonal decay if there exist a rate $M>0$, and a constant $C$ so that $(\delta_x, K \delta_y)\leq C \exp(-M|x-y|)$ for all $x,y \in \Lambda$. 
\end{Def} 

In our case, we note that $-\Delta+V+A_{\textrm{eff}}$ is an exponential decay operator if $W(x-y)$ decays exponentially.

\begin{Lem}[Combes-Thomas estimate]
Let $K$ be an exponential off-diagonal decay operator and $\Sigma$ be its spectrum. Let $\lambda \in \mathbb{C}$ so that $d(\lambda,\Sigma) > 0$. Then there exist $\nu>0$ and $C>0$ such that
\begin{equation}\label{CTomas}
(\delta_x,(K-\lambda)^{-1}\delta_y) \leq C e^{-\nu |x-y|}.
\end{equation}
\end{Lem}

\begin{proof} Let $f(z) := e^{-\nu |x-z|}$ we have  

\begin{eqnarray}
(\delta_x,(K-\lambda)^{-1}\delta_y)  & = & \big(\delta_x,f(z)^{-1}f(z)(K-\lambda)^{-1}f(z)^{-1}f(z)\delta_y \big) \nonumber \\
 & = & f(x)^{-1}f(y)\big(\delta_x,f(z)(K-\lambda)^{-1}f(z)^{-1}\delta_y\big) \nonumber \\
 & = & e^{-\nu |x-y|}\big(\delta_x,(f(z)K f(z)^{-1}-\lambda)^{-1}\delta_y\big). \label{EineCT}
\end{eqnarray} 
Then we have
\begin{equation}
\big((f(z)Kf(z)^{-1}-K)u\big)(n)= \sum_{m\in \Lambda} (e^{{-\nu(|x-n|-|x-m|)}}-1)K(m,n)u(m),
\end{equation}
so
\begin{equation*}
||(f(z)Kf(z)^{-1}-K)u|| \leq C \sum_{m\in\Lambda} (e^{{\nu(|n-m|)}}-1)e^{-M(n-m)}|u(m)| < \infty ,
\end{equation*}
and therefore
\begin{equation*}
||(f(z)Kf(z)^{-1}-H)|| \leq C ||(e^{{\nu |x|}}-1)e^{-M|x|}||_{\ell^1}.
\end{equation*}
Because of the dominated convergence theorem, this converges to $0$ with $\nu$ going to $0$. So there exist $\nu>0$ so that $||(f(z)Kf(z)^{-1}-H)||<d(\lambda, \Sigma)$ and 
\begin{eqnarray*}
d\big(\lambda, \Sigma(f(z)Kf(z)^{-1})\big) \geq d(\lambda, \Sigma)-||f(z)K f(z)^{-1}-H|| > 0.
\end{eqnarray*}
So there exists $C'$ such that
 \begin{equation*}
||(f(z)Kf(z)^{-1}-\lambda)^{-1}|| \leq \frac{1}{d\big(\lambda, \Sigma(f(z)Kf(z)^{-1})\big)} = C'
\end{equation*}
and using \eqref{EineCT} we find
\begin{equation*}
(\delta_x,(K-\lambda)^{-1}\delta_y)\leq C'e^{-\nu|x-y|}.
\end{equation*}

\end{proof}
 
 \subsection{Local influence}
We prove here Theorem \ref{TheoLocal}. We use again the map defined in Section \ref{HamiltonianConstruction}.
\begin{equation}
F_V(\gamma) = \frac{1}{2i\pi}\oint_{\mathcal{C}} (-\Delta+V+A_{\textrm{eff}}(\gamma)-z)^{-1} dz, 
\end{equation}
where $\mathcal{C}$ is the loop enclosing the whole part of the spectrum below the middle of the gap $\mu$. We will denote by $\gamma_{min}(V)$ the solution of the system given by Theorem \ref{TheoExistence} and recall that $F_V(\gamma_{min}(V))=\gamma_{min}(V)$. 

\begin{proof}[Proof of Theorem \ref{TheoLocal}] Because we can write $\delta V = \sum_{k=1}^K \delta V /K$ and apply our theorem $K$ times we can suppose $\delta V$ arbitrary small. We are looking for the new fixed point for $F_{V+\delta V}$ which is the limit of $(F_{V+\delta V})^n(\gamma)$. We start from $\gamma=\gamma_{min}(V)$ and remark that
\begin{equation}
(\gamma_{min}(V+\delta V)-\gamma)=\lim_{n\rightarrow \infty} \Big((F_{V+\delta V})^n(\gamma)-\gamma \Big) = \sum_{n=0}^\infty \Big((F_{V+\delta V})^{n+1}(\gamma)-(F_{V+\delta V})^n(\gamma)\Big).
\end{equation}

\paragraph{Step 1.} We evaluate the first term of the sum $F_{V+\delta V}(\gamma)-\gamma = F_{V+\delta V}(\gamma)-F_{V}(\gamma)$, as follows
\begin{eqnarray*}
F_{V+\delta V}(\gamma)-F_{V}(\gamma) & = & \frac{1}{2i \pi}\oint_{\mathcal{C}} \frac{1}{H_{\textrm{eff}}(\gamma)+\delta V-z}-\frac{1}{H_{\textrm{eff}}(\gamma)-z} dz\\
 & = & \frac{1}{2i \pi} \oint_{\mathcal{C}} \frac{1}{H_{\textrm{eff}}(\gamma)-z} \Big(\delta V \sum_{k\leq 0}\big((H_{\textrm{eff}}(\gamma)-z)^{-1}\delta V\big)^k \Big)\frac{1}{H_{\textrm{eff}}(\gamma)-z} dz \\
 & = & \frac{1}{2i \pi} \oint_{\mathcal{C}} \frac{1}{H_{\textrm{eff}}(\gamma)-z} B(\delta V) \frac{1}{H_{\textrm{eff}}(\gamma)-z} dz , 
\end{eqnarray*}
where 
\begin{equation}
B(\delta V)=\delta V\sum_{k\leq 0}\big((H_{\textrm{eff}}(\gamma)-z)^{-1}\delta V\big)^k.
\end{equation}
This sum converges as soon as $\delta V < \frac{G}{6}$. Note that if the support of $\delta V$ has a bounded support $\Omega$ so has $B(\delta V)$. We will use the Combes-Thomas estimate (\ref{CTomas}) for $y_1,y_2$ outside $\Omega$. Remark that if we choose the loop $\mathcal{C}$ correctly, $\nu$ does not depend on $z$ but only on the size on the gap. We find
\begin{eqnarray*}
|(F_{V+\delta V}(\gamma))(y_1,y_2)- (F_{V}(\gamma))(y_1,y_2)| & = & \oint_{\mathcal{C}} (y_1,\frac{1}{H_{\textrm{eff}}(\gamma)-z} B(\delta V) \frac{1}{H_{\textrm{eff}}(\gamma)-z}y_2) dz \\
 & \leq & \frac{1}{2i \pi}\oint_{\mathcal{C}} C^2 \sum_{x,x'\in \Omega} e^{-\nu |y_1-x|}(x,B(\delta V)x')e^{-\nu |x'-y_2|} dz \\
 & \leq & C^2 |\mathcal{C}| e^{-\nu d(y_1,\Omega)} \sum_{x,x'\in \Omega} |(x,B(\delta V)x')e^{-\nu |x'-y_2|}| \\
 & \leq & C'||\delta V|| e^{-\nu d(y_1,\Omega)}
\end{eqnarray*}
where $C'$ is just a constant. With $W$ small enough, there exists $\tau < 1$ so that 
\begin{equation}
[A_{\textrm{eff}}(F_{V+\delta V}(\gamma))-A_{\textrm{eff}}(\gamma))](y_1,y_2) \leq \tau ||\delta V|| e^{-\nu d(y_1,\Omega)}.
\end{equation}

\paragraph{Step 2.} We evaluate the remainder of the sum. We repeat the previous argument with  
\begin{equation}
\delta A_{\textrm{eff}} = A_{\textrm{eff}}(F_{V+\delta V}^n(\gamma))-A_{\textrm{eff}}((F_{V+\delta V}^{n-1}(\gamma))
\end{equation} instead of $\delta V$. There is only one little difference: $B$ does not have a bounded support any more but still has an off-diagonal exponential decay (proved by iteration with constant $\nu>\nu'>0$). We just check this does not bring more difficulties: 
\begin{align*}
&|(F_{{V+\delta V}})^{n}(\gamma))(y_1,y_2)-(F_{V+\delta V})^{n-1}(\gamma)(y_1,y_2)|  \\
&  \qquad\leq  C^2 \sum_{x,x'\in \Omega} e^{-\nu |y_1-x|}(x,B(\delta A_{\textrm{eff}})x')e^{-\nu |x'-y_2|} dy \\
 & \qquad \leq  C^2 ||\delta A_{\textrm{eff}}|| \sum_{x,x'\in \mathbb{Z}^2} e^{-\nu|y_1-x|} e^{-\nu' d(x,\Omega)}e^{-\nu'  d(x',\Omega)}e^{-\nu |x'-y_2|} \\
&  \qquad\leq  C_2 ||\delta A_{\textrm{eff}}|| e^{-\nu' d(y_1,\Omega)}  \sum_{x,x'\in \mathbb{Z}^2} e^{\nu'd(y_1,\Omega)-\nu |y_1-x|-\nu'd(x,\Omega))}e^{-\nu  d(x',\Omega)}e^{-\nu |x'-y_2|} \\ 
 & \qquad \leq  C_2' ||(F_{{V+\delta V}}^{n-1}(\gamma))-F^{n-2}_{V+\delta V}(\gamma)|| e^{-\nu' d(y_1,\Omega)},
\end{align*} where $C_2'$ is just another constant. We can conclude by iteration that
\begin{equation}
|A_{\textrm{eff}}(F^n_{V+\delta V}(\gamma))(y_1,y_2)-A_{\textrm{eff}}(F^{n-1}_{V}(\gamma))(y_1,y_2)| \leq ||\delta V||\tau^n e^{-\nu' d(y_1,\Omega)},
\end{equation}
and so
\begin{align*}
& |\big( A_{\textrm{eff}}(\gamma_{min}(V+\delta V))-A_{\textrm{eff}}(\gamma_{min}(V))
\big)(y_1,y_2)| \\ & \qquad = \lim_{n\rightarrow\infty} | A_{\textrm{eff}}(F_{V+\delta V}^n(\gamma)) (y_1,y_2)-A_{\textrm{eff}}(\gamma)(y_1,y_2)| \\ & \qquad \leq   ||\delta V|| \frac{1}{1-\tau} e^{-\nu' d(y_1,\textrm{supp}(\delta V))},
\end{align*}
as we wanted.
\end{proof}

\section{Wegner estimate: proof of Theorem \ref{TheoWegner}}

\label{SWegner}

In this section, we prove the Wegner-type estimate in Theorem \ref{TheoWegner}.

The idea of the proof is the following. At first sight we do not have any idea of how $A_{\textrm{eff}}$ looks like. But because of the gap, if
\begin{equation}
\gamma=\mathds{1}_{\leq \mu}\big(-\Delta+V+A_{\textrm{eff}}(\gamma)\big)
\end{equation}
then
\begin{equation}\label{Eoffset}
\gamma=\mathds{1}_{\leq (\mu+\alpha) }\big(-\Delta+V+A_{\textrm{eff}}(\gamma)\big)=\mathds{1}_{\leq (\mu) }\big(-\Delta+V-\alpha \mathds{1}_{\mathbb{Z}^d}+A_{\textrm{eff}}(\gamma)\big) 
\end{equation}
for any $\alpha\in \mathbb{R}$ with $2 |\alpha|$ smaller than the gap. So if we could add $\alpha$ to the random potential with $\alpha$ a smooth random variable, in this case every eigenvalue of $H_{min}$ would just be offset by $\alpha$ no matter what the non linear part is and we are done. In our case, we will make the change of variable $\big(V_\omega(x)\big)\rightarrow (\alpha=\frac{1}{|\Lambda|}\sum_{x\in\Lambda} V_\omega(x), V_\omega(x)-\alpha)$ for $x\in \Lambda$ and we expect that the conditional density of $\alpha$ is smooth enough and that the change induced to $\gamma$ is small.

\begin{proof}[Proof of Theorem \ref{TheoWegner}] Let $\Lambda=\Lambda_L(n)$ be the cube in $\mathbb{Z}^d$ of size $L$ with its center in $n$ and $\Lambda_{2L}(n)$ the cube twice bigger. Because of \eqref{Eoffset} and \eqref{ELocal} 
\begin{eqnarray*}
||\frac{d}{d\alpha}\big((\gamma_{min})(V+\alpha \mathds{1}_{\Lambda_{2L}(n)}\big)_{|\Lambda_L(n)}||& = & || \frac{d}{d\alpha}\big((\gamma_{min})_{|\Lambda_L(n)}(V+\alpha \mathds{1}_{\mathbb{Z}^d} -\alpha \mathds{1}_{\Lambda_{2L}(n)^c}\big) || \\
 & = & || \frac{d}{d\alpha}\big((\gamma_{min})_{|\Lambda_L(n)}(V-\alpha \mathds{1}_{\Lambda_{2L}(n)^c}\big) || \\
 & \leq & C e^{-\nu d\big(\Lambda_L(n),\Lambda_{2L}(n)^c\big)}\\
 & \leq & C e^{-\nu L}.
\end{eqnarray*}
We suppose that $L$ is large enough so that $2||W||_{\ell^1}Ce^{-\nu L} \leq 1/2$ and  we obtain that
\begin{equation}\alpha\rightarrow ||\mathds{1}_{\Lambda_L(n)}A_{\textrm{eff}}(V+\alpha \mathds{1}_{\Lambda_{2L}(n)})\mathds{1}_{\Lambda_L(n)}|| 
\end{equation} is $1/2$ Lipschitz. Under this hypothesis, for any $\lambda_i(\alpha)$ eigenvalue of 
\begin{equation*}
\Big(-\Delta +V+\alpha \mathds{1}_{\Lambda_{2L}(n)}+A_{\textrm{eff}}\big(\gamma_{min}\big(V+\alpha \mathds{1}_{\Lambda_{2L}(n)})\big)\Big)_{|\Lambda_L(n)},
\end{equation*}
 we have
\begin{equation} \label{Eeigcroissant}
\frac{d}{d\alpha} \lambda_i(\alpha) \geq 1-||\frac{d}{d\alpha}(\mathds{1}_{\Lambda_L(n)}A_{\textrm{eff}}(V+\alpha \mathds{1}_{\Lambda_{2L}(n)})\mathds{1}_{\Lambda_L(n)})|| \geq \frac{1}{2}. 
\end{equation}
 
Let $\lambda\in \mathbb{R}$ and $\epsilon > 0$. Let $D_0=\{d_1<d_2<...<d_k\}$ be so that $\rho(s)$ is Lipschitz on $]d_n,d_{n+1}[$. Let $f$ and $\delta$ be two positive functions that will be chosen later. We define the following events
\begin{equation}
O_x := \{\omega :  \forall y \text{ such that } |y-V_\omega(x)|<\delta(\epsilon) : \rho(y)>f(\epsilon)\text{ }  \mbox{, and } d(V_\omega(x),D_0)>\delta(\epsilon)\}
\end{equation}
for any $x\in \Lambda_{2L}(n)$. We now estimate
\begin{align}
& \mathbb{P}\Big(d\big(\sigma[(H_{min})_{|\Lambda_L(n)}], \lambda\big)<\epsilon \Big) \nonumber \\
& \qquad \leq \mathbb{P}\Big(\cup_{ x\in \Lambda_{2L}(n)} O_x^c \Big)+\mathbb{P}\Big(\cap_{x\in \Lambda_{2L}(n)} O_x \cap d\big(\sigma[(H_{min})_{|\Lambda_L(n)}], \lambda\big)<\epsilon \Big). \label{Eine1}
\end{align}
We will deal with each term separately. Starting with the left term, we erase the indices because the probability does not depend of the position and argue as follows:
\begin{align}
\mathbb{P}\Big(\cup_{ x\in \Lambda_{2L}(n)} O_x^c \Big) & \leq |\Lambda_{2L}(n)|\mathbb{P}(O_x^c) \nonumber
\\ & =2^d|\Lambda|\mathbb{P}(O^c) \nonumber
\\ & \leq 2^d|\Lambda|\big(\mathbb{P}(d(V_\omega,D_0)<\delta(\epsilon))+\mathbb{P}(V_\omega : \exists y : |y-V_\omega|<\delta(\epsilon), \rho(y)<f(\epsilon ))\big) \nonumber
\\ & \leq 2^d|\Lambda|\big(\mathbb{P}(d(V_\omega,D_0)<\delta(\epsilon))+\mathbb{P}(V_\omega : \rho(V_\omega)<f(\epsilon)+\delta(\epsilon) ||\rho'||_\infty) \big) \nonumber
\\ & \leq 2^d|\Lambda|  \Big(||\rho||_\infty 2 \delta(\epsilon) \#|D_0|+(\delta(\epsilon) ||\rho'||_\infty + f(\epsilon) )|\text{supp}(\phi)|\Big). \label{Eine2}
\end{align}
The right term in \eqref{Eine1} can be estimated by introducing the mean and the resolvent, using \begin{equation}
\mathds{1}_{[\lambda-\epsilon, \lambda+\epsilon]}(\alpha)\leq \frac{2 \epsilon^2 }{((\lambda-\alpha)^2+\epsilon^2)}=2 \epsilon \Im (\frac{1}{\lambda-\alpha+i\epsilon}).
\end{equation}
We simplify a bit the notation using $\cap_x O_x$ instead of $\cap_{x\in \Lambda_{2L}(n)} O_x$. We get
\begin{align*}
\mathbb{P}\Big(\cap_x O_x \cap d\big(\sigma[(H_{min})_{|\Lambda_L(n)}], \lambda\big)<\epsilon \Big) & = \mathbb{E}\Big[\mathds{1}_{d(\sigma(H^\Lambda_{min}),\lambda)<\epsilon}\mathds{1}_{\cap_x O_x}\Big] 
\\ & \leq 2\epsilon \mathbb{E}\Big[\Im  \Big( \Tr[\big( (H_{min})_{|\Lambda_L(n)}-\lambda+i\epsilon\big)^{-1}]\Big)\mathds{1}_{\cap_x O_x}\Big].
\end{align*}
We now make a change of variable for $V_\omega(x) \in \Lambda_{2L}(n)$ : 
\begin{equation}
\big(V_\omega (x)\big)\rightarrow \Big(\alpha=\frac{1}{2^d|\Lambda|}\sum_{x\in \Lambda_{2L}(n)} V_\omega(x), V_\omega(x)-\alpha\Big).
\end{equation}
We write $\tilde{V}_\omega(x)=V_\omega(x)-\alpha$ and $\xi_{\tilde{V}} (\alpha)$ for the conditional random density of the mean knowing $\tilde{V}$. 
We first integrate over $\alpha$, then over $\tilde{V}$ (we denote the expectation by $\mathbb{E}_{\tilde{V}}$): 
\begin{align*}
& \mathbb{P}\Big(\cap_x O_x\cap d\big(\sigma[(H_{min})_{|\Lambda_L(n)}], \lambda\big)<\epsilon \Big) \\
 & \qquad \leq \mathbb{E}_{\tilde{V}}\Big[2\epsilon \int \Im \Big( Tr[\big((H_{min})_{| \Lambda_L}-\lambda+i\epsilon\big)^{-1}]\Big) \mathds{1}_{\cap_x O_x}\xi_{\tilde{V}_i}(\alpha)  d\alpha \Big] 
\\ & \qquad \leq  2|\epsilon||\mathbb{E}_{\tilde{V}}\Big[\sum_{\lambda_i\in \sigma\big((H_{min})_{| \Lambda_L}\big)} \int \Im (\frac{1}{\big(\lambda_i(\alpha)-\lambda+i\epsilon\big)})\mathds{1}_{\cap_x O_x}\xi_{\tilde{V}_i}(\alpha)  d\alpha\Big].
\end{align*}
We estimate the integral with a change of variable $\alpha'=(\lambda_i(\alpha)-\lambda)$, $d\alpha'=(\frac{d}{d\alpha}\lambda_i)d\alpha$. Recall that 
\begin{equation}
\int\frac{\epsilon}{(\alpha')^2+\epsilon^2} \frac{\xi_{\tilde{V}}(\alpha)\mathds{1}_{\cap_x O_x}}{\frac{d}{d\alpha}\lambda_i} d\alpha'  \leq \pi \sup\big[\frac{|\xi_{\tilde{V}}(\alpha)\mathds{1}_{\cap_x O_x}}{\frac{d}{d\alpha}\lambda_i} \big].
\end{equation} So we have
\begin{align*}
 \mathbb{P}\Big(\cap_x O_x \cap d\big(\sigma[(H_{min})_{|\Lambda_L(n)}], \lambda\big)<\epsilon \Big) & \leq 2 \pi|\epsilon||\Lambda_L|\mathbb{E}_{\tilde{V}}\Big[ \sup\big[\frac{\xi_{\tilde{V}}(\alpha)\mathds{1}_{\cap_x O_x}}{\frac{d}{d\alpha}\lambda_i}\big] \Big]. 
\end{align*}
Finally, because of \eqref{Eeigcroissant}, $\frac{d}{d\alpha}\lambda_i>1/2$ and we get
\begin{align} \label{Eine4}
\mathbb{P}\Big(\cap_x O_x \cap d\big(\sigma[(H_{min})_{|\Lambda_L(n)}], \lambda\big)<\epsilon \Big) &\leq 4 \pi |\epsilon| |\Lambda_L|\sup[\xi_{\tilde{V}_i}\mathds{1}_{\cap_x O_x}].
\end{align}
From now on, it is enough to have an estimate on $\xi_{\tilde{V}}\mathds{1}_{\cap_x O_x}$. A computation gives
\begin{equation}
\xi_{\tilde{V}_i} (\alpha) d\alpha=\mathbb{P}\big(\sum_x V_\omega(x) \in [\alpha, \alpha+ d \alpha]|\tilde{V}\big) =\frac{\prod_{i\in\Lambda}\rho(\tilde{V}_i+\alpha)}{\int\prod_{i\in\Lambda}\rho(\tilde{V}_i+\alpha') d\alpha'} d\alpha .
\end{equation}
Let $\alpha_0 \in \mathbb{R}$. If we do not have $\tilde{V}_\omega(x)+\alpha_0 \in O_x$ for all $x$, then $\xi_{\tilde{V}}(\alpha_0)\mathds{1}_{\cap_x O_x}=0$ and we have nothing else to do. So we can assume that  
$\forall \alpha, |\alpha-\alpha_0| < \delta(\epsilon) \Rightarrow \rho(\tilde{V}(x)+\alpha)>f(\epsilon)$ for all $x\in \Lambda_{2L}(n)$ and all $|\alpha-\alpha_0|<\delta(\epsilon)$. So we have
\begin{equation}
\frac{\frac{d}{d\alpha}\prod_{x \Lambda_{2L}} \rho(\tilde{V}+\alpha)}{\prod_{x\in \Lambda_{2L}} \rho({V}+\alpha)} = \sum_{x\in \Lambda_{2L}}\frac{\rho'(\tilde{V}_\omega(x)+\alpha)}{\rho(\tilde{V}_\omega(x)+\alpha)} \leq 2^d|\Lambda_L| \frac{||\rho'||_{\infty}}{f(\epsilon)} .
\end{equation}
From this differential equation we get
\begin{align*}
& \prod_{x\in \Lambda_{2L}} \rho(\tilde{V}_\omega(x)+\alpha) \geq \exp\Big(-(|\alpha-\alpha_0|)2^d|\Lambda_L|\frac{||\rho'||_{\infty}}{f(\epsilon)}\Big)\prod_{x \in \Lambda_{2L}} \rho(\tilde{V}_\omega(x)+\alpha_0)
\end{align*}
and, after integrating, 
\begin{equation}
\int_{\alpha_0-\delta(\epsilon)}^{\alpha_0+\delta(\epsilon)} \prod\rho(\tilde{V}_\omega(x)+\alpha) d\alpha \geq  \big[1-\exp\big(-|\delta(\epsilon)|2^d |\Lambda_L|\frac{||\rho'||_{\infty}}{f(\epsilon)}\big)\big]\frac{f(\epsilon)}{2^d|\Lambda_L| ||\rho'||_{\infty}}\prod\rho(\tilde{V}_\omega(x)+\alpha_0). \nonumber
\end{equation} Therefore
\begin{equation}
 \frac{\prod\rho(\tilde{V}_\omega(x)+\alpha_0)}{\int\prod\rho(\tilde{V}_\omega(x)+\alpha) d\alpha}  \leq \frac{2^d|\Lambda_L| ||\rho'||_{\infty}}{f(\epsilon)(1-\exp(-(|\delta(\epsilon)|)2^d|\Lambda_L|\frac{||\rho'||_{\infty}}{f(\epsilon)})) }\nonumber
\end{equation} and hence we have
\begin{equation}
 \frac{\prod\rho(\tilde{V}_\omega(x)+\alpha_0)}{\int\prod\rho(\tilde{V}_\omega(x)+\alpha) d\alpha}  \leq 2 \max(\frac{2^d|\Lambda_L| ||\rho'||_{\infty}}{f(\epsilon)},\frac{1}{\delta(\epsilon)}).
\end{equation}
We finally obtain
\begin{equation}
\xi_{\tilde{V}_i}\mathds{1}_{\cap_x O_x} \leq 2 \max(\frac{2^d|\Lambda_L| ||\rho'||_{\infty}}{f(\epsilon)},\frac{1}{\delta(\epsilon)}), \label{Eine3} 
\end{equation} for all $\hat{V}$ and all $\alpha_0$.

To conclude, putting \eqref{Eine1}, \eqref{Eine2}, \eqref{Eine4} and \eqref{Eine3} together, we have
\begin{align*}
& \mathbb{P}\Big(d\big(\sigma[(H_{min})_{|\Lambda_L(n)}], \lambda\big)<\epsilon \Big)  \\
& \qquad \leq 2^{d}|\Lambda| \Big[||\rho||_\infty \delta(\epsilon) \#|D_0|+\big(\delta(\epsilon) ||\rho'||_\infty + f(\epsilon) \big)*|\text{supp}(\rho)| + 8\pi\epsilon. \max(\frac{|\Delta| ||\rho'||_\infty}{f(\epsilon)}, \frac{1}{\delta(\epsilon)}\Big],
\end{align*}
from which we can conclude \eqref{EWegner} choosing  $f(\epsilon) = \sqrt{\epsilon}/|\text{supp}(\rho)|^{-3/2}$ and $\delta(\epsilon)= \sqrt{\epsilon}*|\text{supp} (\rho)|^{1/2}$. Hence
\begin{align*}
& \mathbb{P}\Big(d\big(\sigma[(H_{min})_{|\Lambda_L(n)}], \lambda\big)<\epsilon \Big)  \\
& \qquad \leq C |\Lambda| \Big[|\text{supp}(\rho)|^{-1/2}+||\rho||_\infty |\text{supp}(\rho)|^{1/2}+ ||\rho'||_\infty |\text{supp}(\rho)|^{3/2}\Big],
\end{align*}
where the constant $C$ only depends on the cardinal of $D_0$.
\end{proof}

In particular, with the change $V\rightarrow \ell V $, we deduce that
\begin{equation} \mathbb{P}\Big(d\big(\sigma[(H_{min})_{|\Lambda_L(n)}], \lambda\big)<\epsilon \Big) \rightarrow 0 
\end{equation} 
when $\ell\rightarrow\infty$.

\section{Multiscale analysis}

\label{SMulti}
We will now start the proof of the multiscale analysis. There will be very little differences with the proof we can found in \cite[Part 10]{disertori2008random} and we will follow the method exposed there step by step. But because it is a more general case, we have written the proof again. 

\subsection{The setting}

For any operator $K$ with off-diagonal exponential decay and $\Lambda\subset \mathbb{Z}^d$, we define a border operator $\Gamma$ by
\begin{equation*}
\Gamma_{K,\Lambda}(x,y)=\begin{cases}
K(x,y) \quad \mbox{ if $(x \in \Lambda$ and $y \notin \Lambda)$ or ($y\in \Lambda$ and $x\notin \Lambda)$} \\
0 \quad \mbox{ otherwise.}
\end{cases}
\end{equation*}

The following proposition is a form of the Schur complement formula.
\begin{Prop}
Let $K$ with off-diagonal exponential decay, $\Lambda$ a box of size $L$, and $\lambda \in \mathbb{C}\setminus\mathbb{R}$. Then
\begin{equation}\label{hufre}
(K-\lambda)^{-1}(x,y)=-\sum_{u\in\Lambda,v\notin \Lambda} (K^{\Lambda}-\lambda)^{-1}(x,u)\Gamma_{K,\Lambda}(u,v)(K-\lambda)^{-1}(v,y)
\end{equation}
for any $x\in \Lambda$ and any $y \notin \Lambda$, where
\begin{equation*}
K^\Lambda(x,y)=\begin{cases}
K(x,y) \quad \mbox{ if $x \in \Lambda$ and $y \in \Lambda$} \\
0 \quad \mbox{ otherwise}
\end{cases}
\end{equation*}
is the restriction of $K$ to $\Lambda$.
\end{Prop}

\begin{proof}
We can divide $K$ into the following three parts 
\begin{equation}
K=K^{\Lambda}+\Gamma_{K,\Lambda}+K^{\Lambda^c}
\end{equation}
where $\Lambda^c$ is the complement of $\Lambda$. We here use the resolvent formula
\begin{equation}
(K-\lambda)^{-1}=(K^\Lambda-\lambda)^{-1}+(K^{\Lambda^c}-\lambda)^{-1}-
\big((K^{\Lambda}-\lambda)^{-1}+(K^{\Lambda^c}-\lambda)^{-1}\big)\Gamma_{K,\Lambda}(K-\lambda)^{-1} .
\end{equation}
Just remark now that $(K^\Lambda-\lambda)^{-1}(x,y)=0$ and $(K^{\Lambda^c}-\lambda)^{-1}(x,y)=0$ if $x\in \Lambda$ and $y\notin \Lambda$.
\end{proof}

We now apply the multiscale method. Let $\Lambda_{L}(n)$ be the box of side length $2L+1$ centered at $n\in \mathbb{Z}^d$. We replace the random potential  $V_\omega$ by an arbitrary constant outside the box $\Lambda_{2L}(n)$ in order to make the mean field Hamiltonian inside $\Lambda_L(n)$ independent of what is happening outside $\Lambda_{2L}(n)$: 
\begin{eqnarray*}
\hat{V}_{\omega}^{\Lambda_L(n)}(x)=\begin{cases} V_{\omega}(x) \mbox{ if x} \in \Lambda_{2L}(n), \\
0 \mbox{ otherwise. }
\end{cases}
\end{eqnarray*}
Recall that $0\in \text{supp}(\mathbb{P}) $. From this potential we can obtain with Theorem \ref{TheoExistence} the minimiser $\gamma_{min}(\hat{V}_{\omega}^{\Lambda_L(n)})$ and the mean-field Hamiltonian $H_{min}(\hat{V}_{\omega}^{\Lambda_L(n)})$. We denote its restriction to $\Lambda_L(n)$ by
\begin{equation*}
\hat{H}(n,L):=\Big(H_{min}(\hat{V}_{\omega}^{\Lambda_L(n)})\Big)_{|\Lambda_L(n)} .
\end{equation*}

We introduce this Hamiltonian because of two properties. First it is independent of what is happening outside $\Lambda_{2L}(n)$. Second, it is a good approximation of $(H_{min})_{|\Lambda_L(n)}$. Indeed, from Theorem \ref{TheoLocal} we have
\begin{equation} \label{EPetiteDifference}
||(H_{min})_{|\Lambda_L(n)}-\hat{H}(n,L) || < D e^{-\nu L} 
\end{equation}
where $D$ does not depend on $n$ and $L$. 

\begin{Def}[L-resonance]
A number $\lambda \in \mathbb{R}$ is called L-resonant for the box $\Lambda_L(n)$ if there exists $A_c \mbox{ with } ||A_c||\leq 2 D \exp(-\nu L)$ and
\begin{equation}
 d\big(\lambda,\sigma[\hat{H}(n,L)+ A_c]\big) \leq \exp(-\sqrt{L}),
\end{equation}
 where $D$ is the constant defined in $\eqref{EPetiteDifference}$
\end{Def}
Remark that our definition of non-resonance is equivalent to
\begin{equation}
 d\big(\lambda,\sigma[\hat{H}(n,L)]\big) > \exp(-\sqrt{L})-2 D \exp(-\nu L).
\end{equation}
We have added the operator $A_c$ in the above definition to handle the difference between $(H_{min})_{|\Lambda_L(n)}$ and $\hat{H}(n,L)$. 
This corresponds to Definition 9.1 in \cite{disertori2008random}.

\begin{Def}[($L,\zeta$,$\lambda$)-good box]
The box $\Lambda_L(n)$ is called an ($L$,$\zeta$,$\lambda$)-good box if 
\begin{enumerate}
\item it is not L-resonant; 
\item for any $x \in \Lambda_{\sqrt{L}}(n)$, $y \notin \Lambda_{L}(n)$ and $A_c$ with $||A_c||\leq 2 D \exp(-\nu L)$, 
\begin{equation}
 \sum_v |\hat{H}(n,L)+A_c-\lambda)^{-1}(x,v)|\text{ }|\Gamma_{\hat{H}(n,L)+A_c,\Lambda_L(n)}(v,y)| \leq \exp(-\zeta |y-x|).
\end{equation}
\end{enumerate}
\end{Def}

\subsection{From a scale to another}

Let $L_0$ be not too small and set $L_k=L_0^{\alpha^k}$ with $1<\alpha<2$.
In this subsection, we prove the following proposition.
\begin{Prop} \label{Pscale}
If the following conditions are satisfied
\begin{enumerate}
\item for any 4 boxes of side length $L_k$ in $\Lambda_{L_{k+1}}(n)$, separated from each other by a distance of at least $2L_k$, there is at least one which is $(L_k,\zeta,\lambda)$-good with $\zeta >20/\sqrt{L_k}$,
\item no box in $\Lambda_{L_{k+1}}(n)$ of side length $4L_k$, $12L_k$ ,$20L_k$ is $L_k$-resonant;
\item the domain $\Lambda_{L_{k+1}}(n)$ is not $L_{k+1}$-resonant,
\end{enumerate}
then the cube $\Lambda_{L_{k+1}}(n)$ is $(L_{k+1},\zeta_{k+1},\lambda$)-good with a decay satisfying $\zeta_{k+1}>20/\sqrt{L}$. 
\end{Prop}

This proposition corresponds to Theorem 10.20 in \cite{disertori2008random}. 
\begin{proof}[Proof of Proposition \ref{Pscale}]
Let $A_c$ be so that $||A_c||\leq 2D e^{-\nu L_{k+1}}$. Let $(\hat{H}(n,L_{k+1}))_{|\Lambda_{L_k}(m)}$ be the restriction to the box $\Lambda_{L_k}(m)$ of $\hat{H}(n,L_{k+1})$, with $\Lambda_k(m)\subset \Lambda_{k+1}(n)$. Because of Theorem \ref{TheoLocal}, we have 
\begin{equation} \nonumber
||(\hat{H}(n,L_{k+1}))|_{\Lambda_{L_k}(m)}-\hat{H}(m,L_k)|| \leq D e^{-\nu L_k} 
\end{equation}
and so 
\begin{equation}
||\hat{H}(n,L_{k+1})+A_c)_{|\Lambda_{L_k}(m)} -\hat{H}(m,L)|| \leq 2D e^{-\nu L_k}
\nonumber
\end{equation}
for $L_k$ big enough. 

Let  $K=\hat{H}(n,L_{k+1})+A_c$ and for simplicity we will just write $\Gamma_{\Lambda}$ instead of $\Gamma_{K,\Lambda}$.  Because of what we have just said, if $\Lambda_{L_k}(m)$ is $(L_k,\zeta,E)$ good then
\begin{equation}
 \sum_v (K^{\Lambda_{L_k}(m)}-\lambda)^{-1}(x,v)|\Gamma(v,y)| \leq \exp(-\zeta |y-x|)
\end{equation} and 
\begin{equation}
||(K^{\Lambda_{jL_k}(m)}-\lambda)^{-1}||\leq \exp(\sqrt{j L_k})
\end{equation}
if $\Lambda_{jL_k}(m)$ is not $jL_k-$resonant for $j=4,12,20$. 

The idea is to use equation (\ref{hufre}) as many times as we want. For any $v$ appearing in the equation \eqref{hufre}, we can define another box $\Lambda(v)$ with $y\notin \Lambda(v)$ and repeat the formula with $v$ instead of $x$. Proceeding this way again and again, we get after iteration   
\begin{align}\label{chain}
& (K-\lambda)^{-1}(x,y) \nonumber\\ &= \sum_{(u_i,v_i)_{i=1..n}} (K^{\Lambda_1}-\lambda)^{-1}(x,u_1)\Gamma_{\Lambda_1}(u_1,v_1)(K^{\Lambda_2}-\lambda)^{-1}(v_1,u_2)\Gamma_{\Lambda_2}(u_2,v_2)...(K-\lambda)^{-1}(v_n,y).
\end{align}
We will write the indices of the sum as a tree $\mathcal{T}$ of chains $\mathcal{X}=\big(u_i,v_i,\Lambda_i\big)_{i\leq n}$ with $v_{i}\in \Lambda_i$ $u_{i+1}\in \Lambda_i$, $v_{i+1}\notin \Lambda_i$ and $ y \notin \Lambda_i$.
We first sum over the $u_i$'s so as to reduce our chains to $\mathcal{X}=\big(v_i,\Lambda_i\big)_{i\leq n}$ and we introduce an upper bound $R_\mathcal{X}$ such that
\begin{equation}
R_\mathcal{X} \geq \sum_{(u_i)_{i=1..n}} |(K^{\Lambda_1}-\lambda)^{-1}(x,u_1)\Gamma_{\Lambda_1} (u_1,v_1)(K^{\Lambda_2}-\lambda)^{-1}(v_1,u_2)\Gamma_{\Lambda_2}(u_2,v_2)\cdots\Gamma(u_n,v_n)|.
\end{equation} Then, Equation \eqref{chain} gives  
\begin{equation}
|(K-\lambda)^{-1}(x,y)|\leq ||(K-\lambda)^{-1}||\sum_{\mathcal{X} \text{ leaves of } \mathcal{T}} R_\mathcal{X}.
\end{equation}
This formula is very general and is valid for any expansion. Different choices for the construction of the tree exit in the literature and we will follow that of \cite{disertori2008random}. The goal is to get at least one good box at each step. The choice of the $\Lambda_i$ and the construction of the tree $\mathcal{T}$ of $\mathcal{X}$ and $R_{\mathcal{X}}$ are made according to the following algorithm. 

We start from $x$ so we define $v_0=x$ and  $R=1$. The choice of $\Lambda_{i+1}$ will depend on $v_i$.
 \begin{itemize}
\item If we get close to the boundary or to $y$, \big($d(v_i, \partial \Lambda_{L_{k+1}}(n))<L_k$ or $d(v_i,y)<L_k$\big) then we stop. The construction of this chain is over and we carry on with the other branches of the tree. 
\item Otherwise
\begin{itemize}
\item if $\Lambda_{L_k}(v_i)$ is an $(L_k,\zeta,E)$-good box then 
\begin{align*}
& R_\mathcal{X}(K-\lambda)^{-1}(v_i,y)  \\  
& \qquad \leq R_\mathcal{X}\sum_{u_{i+1},v_{i+1}}(K^{\Lambda_i}-\lambda)^{-1}(v_i,u_{i+1})\Gamma_{\Lambda_i}(u_{i+1},v_{i+1}) (K-\lambda)^{-1}(v_{i+1},y)\\  
	& \qquad \leq \sum_{v_{i+1}\notin\Lambda_{L'}(v_i)} R_\mathcal{X} \exp(-\zeta |v_{i+1}-v_i|) (K-\lambda)^{-1}(v_{i+1},y)
\end{align*}
so for each $v_{i+1}$ outside $\Lambda_L(v_i)$, we set
\begin{equation} \label{EAlgo1}
R_{\mathcal{X}+v_{i+1}}=R_\mathcal{X} \exp(-\zeta|v_{i+1}-v_i|)
\end{equation}
and carry on the algorithm with the new chain $\mathcal{X}+v_{i+1}$;

\item else if $\Lambda_{L_k}(u_i)$ is not a good box, choose $j = 4$ or $12$ or $L_k$ such that for every $v$ in $\Lambda_{2jL_k}\setminus \Lambda_{jL_k}$, $\Lambda_{L_k}(v)$ is a good box and $\Lambda_{jL_k}$ is not resonant. It is always possible to do this because of the following remark: Either 3 boxes are far away from each other then there are 3 boxes $M_1,M_2,M_3$ of size $4L_{k}$ separated by at least $2L_k$ so that every cube $\Lambda_{L_k}(m)\subset\Lambda_{L_{k+1}}(m)$ whom center is not included in $\cup_{i=1,2,3} M_i$ are $(L_k,\zeta, \lambda)$-good. Or two of them are close and the other is far away then there are two boxes $M_1$ of size $12L_k$ and $M_2$ of size $4L_k$ separated by at least $2L_k$ so that every cube in $\Lambda_{L_k}(m)\subset\Lambda_{L_{k+1}}(m)$ whom center is not included in  $M_i$ are $(L_k,\zeta, \lambda)$-good. Or the three of them are together then there exist one box $M_1$ of size $20L_k$ so that every cube $\Lambda_{L_k}(m)\subset\Lambda_{L_{k+1}}(m))$ whom center is not included in $M_1$ are $(L_k,\zeta, \lambda)$-good. We can assume than the good box decay is smaller than the off diagonal decay parameter $\nu$, $\zeta < \nu$.
We then have

\begin{align*}
& R_\mathcal{X}(K-\lambda)^{-1}(v_i,y) \\
& \quad \leq R_\mathcal{X} \sum_{u_{i+1},v_{i+1}}(K^{\Lambda_{i+1}}-\lambda)^{-1}(v_i,u_{i+1})\Gamma_{\Lambda_{i+1}}(u_{i+1},v_{i+1}) (K-\lambda)^{-1}(v_{i+1},y) \\
& \quad \leq R_\mathcal{X} \sum_{u_{i+1}, v_{i+1} \in \Lambda_{2jL_k}(v_i)}(K^{\Lambda_{i+1}}-\lambda)^{-1}(v_i,u_{i+1})\Gamma_{\Lambda_{i+1}}(u_{i+1},v_{i+1}) (K-\lambda)^{-1}(v_{i+1},y) \\ & \qquad +\sum_{u_{i+1},v_{i+1}\notin \Lambda_{2jL_k}(u_i)}(K^{\Lambda_i}-\lambda)^{-1}(v_i,u_{i+1})\Gamma_{\Lambda_{i+1}}(u_{i+1},v_{i+1}) (K-\lambda)^{-1}(v_{i+1},y)\\
& \quad \leq R_\mathcal{X} \sum_{\substack{ u_{i+1}, v_{i+1} \in \Lambda_{2jL_k}(v_i)\\ u'_{i+1}, v'_{i+1} \notin \Lambda_{L_k}(v_{i+1})}}(K^{\Lambda_{i+1}}-\lambda)^{-1}(v_i,u_{i+1})\Gamma(u_{i+1},v_{i+1}) \times \\ 
& \qquad \qquad \times(K^{\Lambda_{L_k}(v_{i+1})}-\lambda)^{-1}(v_{i+1},u_{i+1}')\Gamma_{\Lambda_{L_k}(v_{i+1}}(u_{i+1}',v_{i+1}')(K-\lambda)^{-1}(v_{i+1}',y) \\ 
& \qquad +\sum_{v_{i+1}\notin \Lambda_{2jL_k}(u_i)} (jL_k)^d C e^{\sqrt{jL_k}} \exp\big(-\nu (|v_{i+1}-v_{i}|-jL_k)\big)(K-\lambda)^{-1}(v_{i+1},y) \\
& \quad \leq R_\mathcal{X} \sum_{v_{i+1}} (2jL_k)^d e^{\sqrt{jL_k}} \exp\big(-\zeta \max(|v_{i+1}-v_{i}|-2jL_k,L_k)\big) (K-\lambda)^{-1}(v_{i+1},y) \\ 
& \qquad +\sum_{v_{i+1}\notin \Lambda_{2jL_k}(u_i)} (jL_k)^d C e^{\sqrt{jL_k}} \exp\big(-\nu (|v_{i+1}-v_{i}|-jL_k)\big) (K-\lambda)^{-1}(v_{i+1},y).
 \end{align*}
 Therefore we can set
\begin{equation} \label{EAlgo2}
R_{\mathcal{X}+v_{i+1}}=R_\mathcal{X} 2 C (2jL_k)^d \exp(\sqrt{jL_k}) \exp\big(-\zeta\max(L_k,|v_{i+1}-v_i|-2jL_k)\big)
\end{equation}
and as previously, we carry on with the new chain $\mathcal{X}+v_{i+1}$.
\end{itemize}
\end{itemize}
\begin{figure}[h!]
\centering
\includegraphics[width = 10cm]{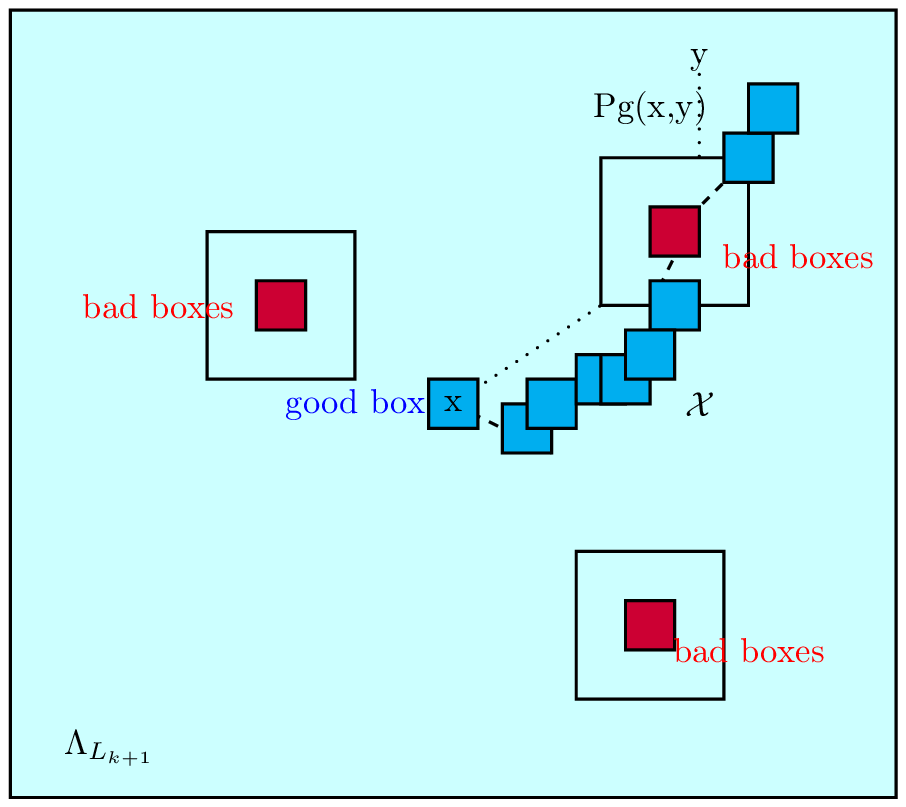}
\caption{Schematic representation a typical chain $\mathcal{X}$ used in the proof of Proposition \ref{Pscale}}
\end{figure}

We have finished the description of the algorithm. We define the "good path length" $P_g(x,v)$ as the minimum length between $x$ and $v$ when any cube $M=\Lambda_{2j L_k}(v)$ containing bad boxes defined in the procedure can be crossed for free. We easily check that \eqref{EAlgo1} and \eqref{EAlgo2} imply that
\begin{equation}
R_{\mathcal{X}} \leq \exp(-\zeta \max(\textrm{length}(\mathcal{X}) L_k, P_g(x,v_i)). 
\end{equation}
for every chain $\mathcal{X}$. From this, our algorithm gives us the following estimate:
\begin{Prop}
\begin{equation}
\sum_{\mathcal{X}} R_\mathcal{X} \leq \frac{1}{1-(L_{k+1})^{d} \exp(-\zeta L_k)} (L_{k+1})^{d \frac{L_{k+1}}{L_k}} \exp(-\zeta (P_g(x,y)).
\end{equation}
\end{Prop}
\begin{proof}
We have
\begin{eqnarray*}
\sum_{(\mathcal{X})} R_\mathcal{X} & =& \sum_{N=1}^\infty \sum_{\mathcal{X},N=\textrm{length}(\mathcal{X})} R_\mathcal{X} \\
 & \leq & \sum_{N=1}^\infty \sum_{\mathcal{X},N=\textrm{length}(\mathcal{X})} \exp(-\zeta max(N L_k, P_g(x,y)) \\
 & \leq & \sum_{N=1}^\infty \sum_{\mathcal{X},N=\textrm{length}(\mathcal{X})} \exp(-\zeta (P_g(x,y)+max(0,N-\frac{L_{k+1}}{L_k})L_k) \\
 & \leq & (L_{k+1})^{d \frac{L_{k+1}}{L_k}} \exp(-\zeta (P_g(x,y) \sum_{N=0}^\infty L_{k+1}^{Nd} \exp(-\zeta N L_k) \\
 & \leq & \frac{1}{1-(L_{k+1})^{d} \exp(-\zeta L_k)} (L_{k+1})^{d \frac{L_{k+1}}{L_k}} \exp(-\zeta (P_g(x,y)),
\end{eqnarray*}
because at each step $i$ there are only $L_{k+1}^d$ possible $v_i$'s.
\end{proof}

The fact that there are only 3 bad boxes implies that the good path length is close to the usual distance $P_g(x,y) \geq |x-y|-20L_k $. We can now conclude. Let $y$ be outside $\Lambda_{L_{k+1}}$. Then
\begin{align*}
& \sum_v |(\hat{H}(n,{L+1})+A_c-\lambda)^{-1}(x,v)||\Gamma(v,y)| \\ & \qquad \leq  \sum_v ||((\hat{H}(n,{L+1})+A_c-\lambda)^{-1}||.\frac{1}{1-1-(L_{k+1})^{d} \exp(-\zeta L_k)} (L_{k+1})^{L_{k+1}/L_{k}} \times \\ & \qquad \qquad\times \exp(-\zeta (|x-v|-20))C \exp(-\nu |v-y|) \\
 & \qquad \leq  C'(L_{k+1})^d\exp(\log(L_k)L_k^{\alpha-1} ) \exp(\sqrt{L_{k+1}})\exp(-\zeta (|x-y|-20)
\end{align*}
In order to conclude, we just remark that $\sqrt{L_k}+\log(L_k)L_k^{\alpha-1}+d \log (L_k)=o(|x-y|)$. So we can choose
\begin{equation}
\zeta_{k+1}=\zeta_k-\frac{\sqrt{L_k}+\log(L_k)L^{\alpha-1}+d \log (L_k)}{L_k},
\end{equation}
which finishes the argument.
\end{proof}

\subsection{The multiscale}

\begin{Theo}Assume that there exists a gap big enough in the spectrum, that the law of potential has a density Lipschitz by part and that $||W||_{L^1}$ is small enough compared to the gap. Then, let $L_0$ be large enough, $\lambda\in \mathbb{R}$ $\zeta >1/\sqrt{L_0}$, $p>2d$ and $ 1< \alpha < 2p/(p+2d)$. If for any cubes $\Lambda_{L_0}(n_0)$, $\Lambda_{L_0}(m_0)$ separated by at least $2 L_0$,

\begin{equation}
\mathbb{P}\big(\exists \lambda\in I : \mbox{$\Lambda_{L_0}(n)$ and $\Lambda_{L_0}(n)$ are not $(L_0,\zeta, \lambda$)-good } \big) \leq \frac{1}{L_0^{2p}},
\end{equation}
then for any $k$, and any cubes $\Lambda_{L_k}(n_k)$, $\Lambda_{L_k}(m_k)$ separated by at least $2 L_k$, we  have
\begin{equation}
\mathbb{P}\big(\exists \lambda\in I : \mbox{ $\Lambda_{L_k}(n_k)$ and $\Lambda_{L_k}(m_k)$ are not $(L_k,\zeta, \lambda$) good }\big) \leq \frac{1}{L_k^{2p}},
\end{equation}
where $L_{k+1}=L_k^\alpha$
\end{Theo}
This theorem is similar to Theorem 10.22 of \cite{disertori2008random}.

\begin{proof} The demonstration is done by iteration using Proposition \ref{Pscale}. Suppose there exists $\lambda \in I$ such that $\Lambda_{L_{k+1}}(n)$ and $\Lambda_{L_{k+1}}(m)$ are not $(L_{k+1},\zeta,\lambda)$-good. Then for each box one of the hypothesis of Proposition \ref{Pscale} fails. So either one of the boxes admits 4 separated bad sub-boxes, or there exists $\lambda$ such that for the two boxes, one of their sub-boxes shows a resonance at $\lambda$. The probability that the hypothesis over the existence of 4 bad boxes is not true can be estimated by iteration. Indeed, 4 cubes means 2 pairs. Because of independence, the 4-cubes probability will be the 2-cubes probability squared and because there are only $L^d$ cubes, hence $L^{4d}$ 4 cubes combinations we can estimate this probability by
\begin{eqnarray*}
&\leq & C L^{4d} \left(\frac{1}{(L_k)^{2p}}\right)^2 \\
&\leq  &C \frac{1}{(L_{k+1})^{4\frac{p}{\alpha}-4d}}\\
&\leq &\frac{1}{4} \frac{1}{L_{k+1}^{2p}}.
\end{eqnarray*}
The last inequality is true for $L_k$ big enough because $2p < 4\frac{p}{\alpha}-4d$ so $\alpha < \frac{2p}{p+2d}$.

The probability of the non resonance hypothesises is controlled by Wegner estimate (\ref{EWegner}) as it is done Theorem 10.22 of \cite{disertori2008random} with $\mathcal{O}(L^d\sqrt{e^{-\sqrt{L_k}}})= o (L^{-2p})$.
\end{proof}

From this and (\ref{EWegner}) we can deduce the following corollaries. The proof can be found again in \cite[Part 9 and Part 11]{disertori2008random} . 

\begin{Cor}In presence of strong disorder, meaning \begin{equation}
\Big(|\text{supp}(\rho)|^{-1/2}+||\rho||_\infty | \text{supp}(\rho)|^{1/2}+ ||\rho'||_\infty |\text{supp}(\rho)|^{3/2}\Big) \nonumber
\end{equation}
small enough, then $H_{min}$ has pure point spectrum and its eigenvectors are localised in space. 
\end{Cor}

Furthermore the Lifshitz tail is not modified too much for $||W||$ very small because  
\begin{equation}
\mathbb{P}\big[d(\sigma[(H_{min})_{|\Lambda}], \lambda)<\epsilon\big] \leq 
\mathbb{P}\big[d(\sigma[H^\Lambda], \lambda)<\epsilon + 2||W||_{\ell^1} \big],
\end{equation}
so we also get the following 
\begin{Cor}
There is $\epsilon$ such that if $||W||_{\ell^1}<\epsilon$ there are small intervals at the edges of the bands of the spectrum of $H_{min}$, where the spectrum is pure point with exponentially decaying eigenvectors.
\end{Cor}

The two results conclude the proof of Theorem \ref{TheoMain}.
\section{Numerical simulations}

\label{Snumeric}
In this section, we present some simple numerical simulations in order to illustrate our theorems. Due to the computational cost we restrict ourselves to the one dimensional case, which however is known to present stronger localisation effects than in higher dimensions. It would be interesting to generalise our simulation to dimension 2 and 3. 

We take $\Lambda=[0,L]$ with $L$ ranging from two hundred to a few thousands lattice sites, $200 \leq L \leq 2000$. The discrete one-dimensional Laplacian is defined in \eqref{Laplacian}. The deterministic potential $V_0$ is 2-periodic: 
\begin{equation}
V_0(n)= \begin{cases} \xi & \text{if $n$ is even} \\ -\xi & \text{if $n$ is odd,} \end{cases}
\end{equation}
where $\xi>0$ is a parameter. The probability of the random potential $V_\omega$ is the uniform law over the interval $[0,\zeta]$ where $\zeta$ is another parameter. For $W(x-y)$, we use a simple next-to-nearest neighbour interaction of the form 
\begin{equation}
 W(x-y)= \begin{cases} q & \text{if $x-y=0$} \\ q/2 & \text{if $|x-y|=1$} \\ q/4  &\text{if $|x-y|=2$ or $3$} \\ 0 & \text{otherwise,}   \end{cases}
 \end{equation} 
where $q$ is another parameter. Our model depends therefore on three parameters $\xi$, $\zeta$ and $q$.
When $\xi>2+\zeta$ the spectrum of the linear Hamiltonian $-\Delta+V$ is composed of two distinct intervals. We then choose $q$ such as to keep a gap in the spectrum and ensure that the map $F$ is contracting (Theorem \ref{TheoExistence}). In this model the particles fill half of the energy states, that is, there are $N=L/2$ particles.  

\subsection{Illustration of Theorems \ref{TheoExistence}, \ref{TheoLocal}, \ref{TheoWegner} and \ref{TheoMain} }

In order to construct the solution $\gamma_{min}$ and the associated mean-field Hamiltonian $H_{min}$, we use the fixed point algorithm employed in the proof of Theorem \ref{TheoExistence}. In Figure \ref{FigureDecroit} we display and confirm the exponential decrease of $||F^{n+1}(\gamma_0)-F^{n}(\gamma_0)||$ with initial condition $\gamma_0=\mathds{1}_{\leq \mu}(-\Delta+V)$.
\begin{figure} [p] 
\centering  
   \includegraphics[width=9cm]{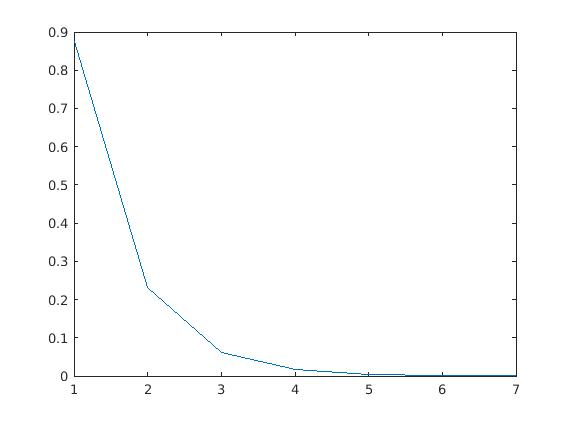} 
   \caption{Values of $||F^{n+1}(\gamma_0)-F^{n}(\gamma_0)||$ with $\gamma_0=\mathds{1}_{\leq \mu}(-\Delta+V)$ and $\xi= 1$, $\zeta= 1$ $q = 2$, $L=500$. \label{FigureDecroit}}
\end{figure}

Next we have tested that adding a dirac at the site 250 to the potential $V$ induces a perturbation in the non linear minimiser, which decays exponentially fast (Theorem \ref{TheoLocal}). In Figure \ref{FigurePertu} we plot the relative density $ \gamma_{min}(V+\delta_{250})(x,x)-\gamma_{min}(V)(x,x)$. 
\begin{figure}[p]
\centering
\includegraphics[width=9cm]{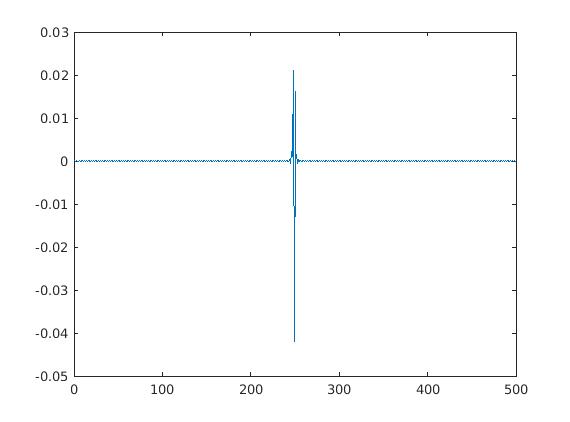} 
\caption{The relative density $ \gamma_{min}(V+\delta_{250})(x,x)-\gamma_{min}(V)(x,x)$ with $\xi= 1$, $\zeta= 1$ $q = 2$, $L=500$. \label{FigurePertu}}
\end{figure}

In Figure \ref{FigureDensi} we have tested the Wegner-type estimate of Theorem  \ref{TheoWegner}, where we have obtained a bound in terms of $\sqrt{\epsilon}$ instead of the usual $\epsilon$ that can be found in the literature \cite[Prop.VIII.4.11]{carmona2012spectral}. In dimension 1, we observe that the usual bound should hold in the nonlinear case but we are unable to prove it so far. 
\begin{figure}[p]
\centering
\includegraphics[width=9cm]{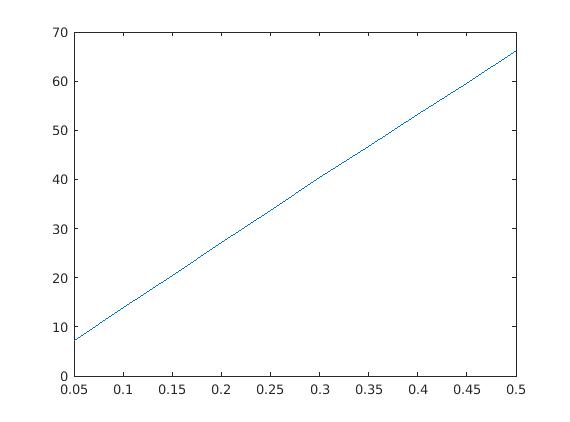}
\caption{Average number of eigenvalues of $H_{min}$ in an interval $I$ according to the size of $I$. We plot $\mathbb{E}(\Tr(\mathds{1}_{[2;2+\epsilon]}(H_{min})))$ in terms of $\epsilon$. The values of the parameters are $\xi=1$, $\zeta=1$ and $q=2$.  \label{FigureDensi}}
\end{figure}

We conclude the illustration of our results with the case of Theorem \ref{TheoMain}. Since we deal with a one-dimensional system all the eigenvectors are localised, even in a regime of parameters which is not covered by the second part of Theorem $\ref{TheoMain}$. This is shown in Figure \ref{FigureLoca}. It is an interesting open problem to prove a stronger localisation result in the one-dimensional Hartree-Fock model.
\begin{figure}[p]
\centering
\includegraphics[width=7cm]{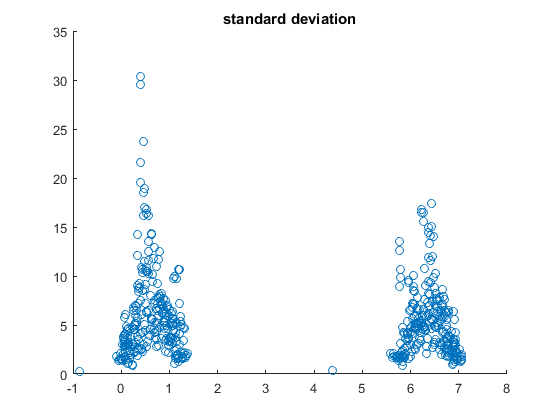}
   \includegraphics[width=7cm]{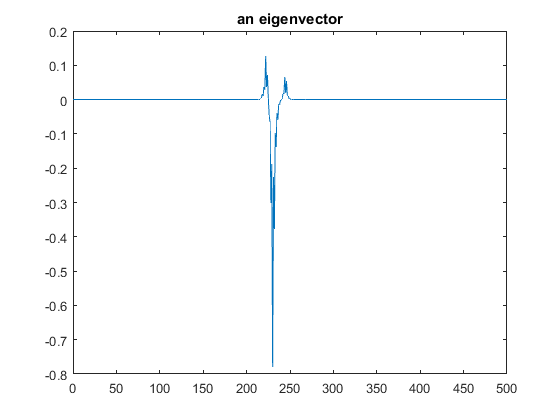} 
   \caption{\label{FigureLoca} Left: Standard deviation of all the eigenvectors of the mean field operator $H_{min}$ in terms of theirs eigenvalue. Here the size of the domain is $L=1000$, hence a value of 10 shows localisation. Right: an eigenvector chosen at random. The values of the parameters are $\xi= 1$, $\zeta= 1$ and $q = 2$.}
\end{figure}    

\subsection{Closing the gap: insulators and metals}

In Figure \ref{FigureClosing2} and \ref{FigureClosing7} we have increased the intensity $q$ of the interaction up to the point where the gap closes. For a large interaction the fixed point algorithm used to construct the solution in Theorem \ref{TheoExistence} does not work. Instead we have used the optimal damping algorithm of \cite{CanBri-00a}  which works perfectly. In general, we observe that the eigenvectors are less localised except at the edges of the spectrum. There is no sign of a phase transition. 

In Figure \ref{FigureSansStructure}, we have erased the gap by choosing $\xi=0$. A small delocalisation phenomenon seems to appear at the Fermi energy $\mu$. 

We hope to be able to understand these phenomena rigorously  in the future.  
\begin{figure} [p] 
\centering  
   \includegraphics[width=10cm]{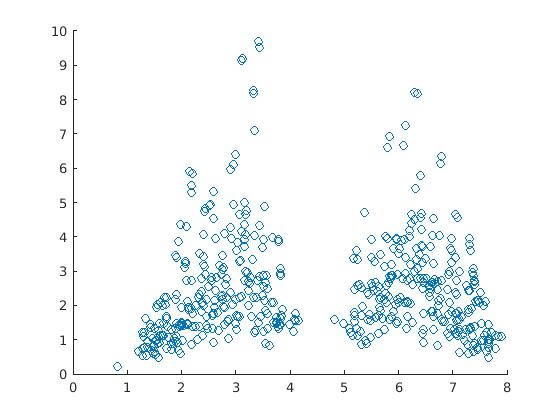} 
\caption{Standard deviation of the eigenvectors in terms of their eigenvalue, for $\xi= 2$, $\zeta= 3$ and $q = 2$.  \label{FigureClosing2}}
\end{figure}

\begin{figure} [p] 
\centering  
   \includegraphics[width=10cm]{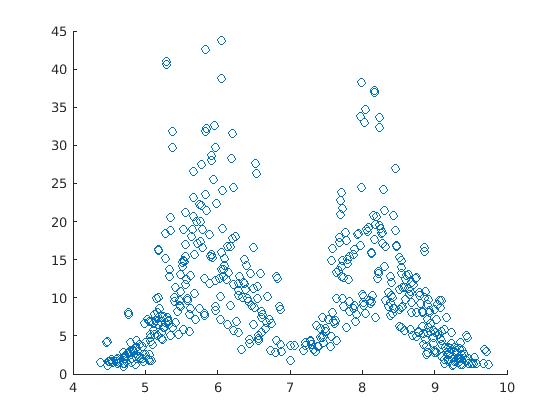} 
	\caption{Standard deviation of the eigenvectors in terms of their eigenvalue, for $\xi= 2$, $\zeta= 3$ and $q = 7$.  \label{FigureClosing7}}
\end{figure} 

\begin{figure} [p] 
\centering  
   \includegraphics[width=10cm]{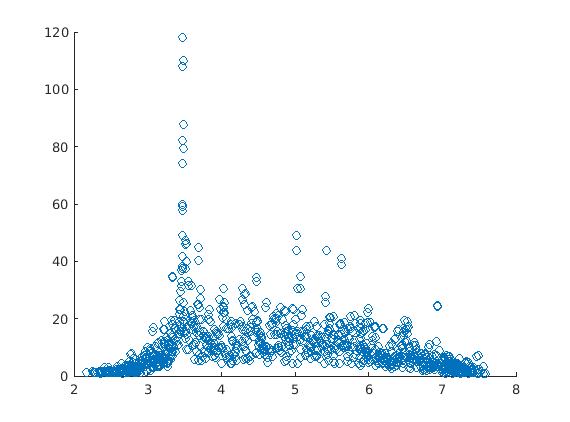} 
	\caption{Standard deviation of the nonlinear eigenvectors in terms of their eigenvalue, with the periodic potential dropped out ($\xi= 0$). Here there are only $L/4$ particles. The Fermi level is $\mu=3.5$ and there is no gap. A small delocalisation seems to appear at the Fermi level. The other parameters are $\zeta= 4$ and $q = 4$.  \label{FigureSansStructure}}
\end{figure} 
 
\subsection{The influence of the periodic potential $V_0$}

The periodic potential seems to have an influence on the localisation phenomena even in the linear case. Our Figure \ref{FigurePeriDep} clearly illustrates that the periodic potential favours localisation. To our knowledge there are very few mathematical results about how a small random potential influences a highly varying periodic system. It is however an important question if we think of the absence of conductivity in ionic crystals. 

\begin{figure}[p]
\centering
\includegraphics[scale=0.5]{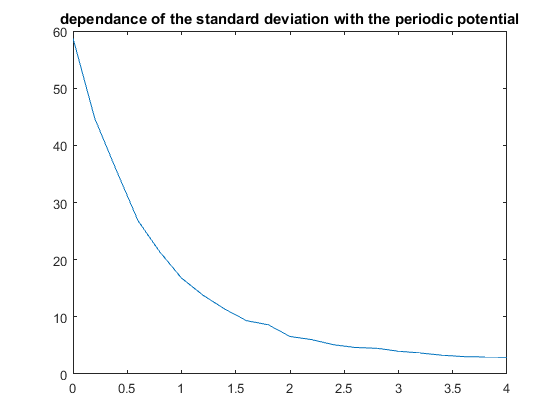}
\caption{Mean of the standard deviation for the eigenvectors of the linear problem ($q=0$) according to the strength $\xi=0,..,4$ of the periodic potential, with a fixed random intensity $\zeta=2$. \label{FigurePeriDep}}
\end{figure}

\paragraph{Acknowledgement} I would like to thank Mathieu Lewin for his help and all the time he gave me during this work.


\newpage

\begin{thebibliography}{10}

\bibitem{aizenman1993}
{\sc M.~Aizenman and S.~Molchanov}, {\em Localization at large disorder and at
  extreme energies: an elementary derivation}, Comm. Math. Phys., 157 (1993),
  pp.~245--278.

\bibitem{2009CMaPh.290..903A}
{\sc M.~{Aizenman} and S.~{Warzel}}, {\em {Localization Bounds for
  Multiparticle Systems}}, Communications in Mathematical Physics, 290 (2009),
  pp.~903--934.

\bibitem{PhysRev.109.1492}
{\sc P.~W. Anderson}, {\em Absence of diffusion in certain random lattices},
  Phys. Rev., 109 (1958), pp.~1492--1505.

\bibitem{Bach-92}
{\sc V.~Bach}, {\em Error bound for the {H}artree-{F}ock energy of atoms and
  molecules}, Commun. Math. Phys., 147 (1992), pp.~527--548.

\bibitem{BacLieLosSol-94}
{\sc V.~Bach, E.~H. Lieb, M.~Loss, and J.~P. Solovej}, {\em There are no
  unfilled shells in unrestricted {H}artree-{F}ock theory}, Phys. Rev. Lett.,
  72 (1994), pp.~2981--2983.

\bibitem{BacLieSol-94}
{\sc V.~Bach, E.~H. Lieb, and J.~P. Solovej}, {\em Generalized {H}artree-{F}ock
  theory and the {H}ubbard model}, J. Statist. Phys., 76 (1994), pp.~3--89.

\bibitem{2006AnPhy.321.1126B}
{\sc D.~M. {Basko}, I.~L. {Aleiner}, and B.~L. {Altshuler}}, {\em {Metal
  insulator transition in a weakly interacting many-electron system with
  localized single-particle states}}, Annals of Physics, 321 (2006),
  pp.~1126--1205.

\bibitem{BasAleAlt-07}
\leavevmode\vrule height 2pt depth -1.6pt width 23pt, {\em Problems of
  Condensed Matter Physics}, International Series of Monographs on Physics,
  Oxford University Press, 2007, ch.~On the problem of many-body localization.

\bibitem{BlaLew-12}
{\sc X.~Blanc and M.~Lewin}, {\em Existence of the thermodynamic limit for
  disordered quantum {C}oulomb systems}, J. Math. Phys., 53 (2012), p.~095209.
\newblock Special issue in honor of E.H. Lieb's 80th birthday.

\bibitem{CanBri-00a}
{\sc {\'E}.~Canc{\`e}s and C.~{Le Bris}}, {\em On the convergence of {SCF}
  algorithms for the {H}artree-{F}ock equations}, M2AN Math. Model. Numer.
  Anal., 34 (2000), pp.~749--774.

\bibitem{Cances2013241}
{\sc E.~Cancès, S.~Lahbabi, and M.~Lewin}, {\em Mean-field models for
  disordered crystals}, Journal de Mathématiques Pures et Appliquées, 100
  (2013), pp.~241 -- 274.

\bibitem{carmona2012spectral}
{\sc R.~Carmona and J.~Lacroix}, {\em Spectral theory of random Schr{\"o}dinger
  operators}, Springer Science \& Business Media, 2012.

\bibitem{2009MPAG...12..117C}
{\sc V.~{Chulaevsky} and Y.~{Suhov}}, {\em {Multi-particle Anderson
  Localisation: Induction on the Number of Particles}}, Mathematical Physics,
  Analysis and Geometry, 12 (2009), pp.~117--139.

\bibitem{disertori2008random}
{\sc M.~Disertori, W.~Kirsch, and A.~Klein}, {\em Random Schrodinger
  Operators}, Panoramas et synth{\`e}ses, Soci{\'e}t{\'e} math{\'e}matique de
  France, 2008.

\bibitem{2015RvMaP..2750010F}
{\sc M.~{Fauser} and S.~{Warzel}}, {\em {Multiparticle localization for
  disordered systems on continuous space via the fractional moment method}},
  Reviews in Mathematical Physics, 27 (2015), pp.~1550010--197.

\bibitem{frohlich1983}
{\sc J.~Fröhlich and T.~Spencer}, {\em Absence of diffusion in the anderson
  tight binding model for large disorder or low energy}, Comm. Math. Phys., 88
  (1983), pp.~151--184.

\bibitem{PhysRevLett.95.206603}
{\sc I.~V. Gornyi, A.~D. Mirlin, and D.~G. Polyakov}, {\em Interacting
  electrons in disordered wires: Anderson localization and low-$t$ transport},
  Phys. Rev. Lett., 95 (2005), p.~206603.

\bibitem{2014arXiv1406.2957I}
{\sc J.~Z. {Imbrie}}, {\em {Multi-Scale Jacobi Method for Anderson
  Localization}}, ArXiv e-prints,  (2014).

\bibitem{KliBootstap}
{\sc A.~Klein and S.~Nguyen}, {\em The bootstrap multiscale analysis for the
  multi-particle anderson model}, Journal of Statistical Physics, 151 (2013),
  pp.~938--973.

\bibitem{KonMosSeiYng-14}
{\sc M.~K\"onenberg, T.~Moser, R.~Seiringer, and J.~Yngvason}, {\em Superfluid
  behavior of a {B}ose-{E}instein condensate in a random potential}, New
  Journal of Physics, 17 (2015), p.~013022.

\bibitem{lahbabi:tel-00873213}
{\sc S.~Lahbabi}, {\em {Mathematical study of quantum and classical models for
  random materials in the atomic scale}}, phd thesis, {Universit{\'e} de Cergy
  Pontoise}, July 2013.

\bibitem{lahbabi:hal-00797094}
\leavevmode\vrule height 2pt depth -1.6pt width 23pt, {\em {The Reduced
  Hartree-Fock Model for Short-Range Quantum Crystals with Nonlocal Defects}},
  {Annales Henri Poincar{\'e}},  (2013), pp.~1--50.

\bibitem{LieSim-77}
{\sc E.~H. Lieb and B.~Simon}, {\em The {H}artree-{F}ock theory for {C}oulomb
  systems}, Commun. Math. Phys., 53 (1977), pp.~185--194.

\bibitem{Lions-87}
{\sc P.-L. Lions}, {\em Solutions of {H}artree-{F}ock equations for {C}oulomb
  systems}, Commun. Math. Phys., 109 (1987), pp.~33--97.

\bibitem{2015arXiv151205282S}
{\sc R.~{Seiringer} and S.~{Warzel}}, {\em {Decay of correlations and absence
  of superfluidity in the disordered Tonks-Girardeau gas}}, ArXiv e-prints,
  (2015).

\bibitem{SeiYngZag-12}
{\sc R.~Seiringer, J.~Yngvason, and V.~A. Zagrebnov}, {\em Disordered
  {B}ose-{E}instein condensates with interaction in one dimension}, Journal of
  Statistical Mechanics: Theory and Experiment, 2012 (2012), p.~P11007.

\bibitem{Solovej-03}
{\sc J.~P. Solovej}, {\em The ionization conjecture in {H}artree-{F}ock
  theory}, Ann. of Math. (2), 158 (2003), pp.~509--576.

\end{thebibliography}

\end{document}